\documentclass[12pt]{amsart}
\usepackage{amssymb}
\usepackage{amsmath}
\usepackage{amscd}
\usepackage[dvips]{color}
\newtheorem{theorem}{Theorem}[section]

\newtheorem{lemma}[theorem]{Lemma}
\newtheorem{coro}[theorem]{Corollary}
\newtheorem{prop-def}{Proposition-Definition}[section]

\begin{document}

\title{ Constructing  $3$-Lie algebras}

\author{Ruipu  Bai}
\address{College of Mathematics and Computer Science,
Hebei University, Baoding 071002, P.R. China}
\email{bairuipu@hbu.cn}

\author{Yong Wu }
\address{College of Mathematics and Computer Science,
Hebei University, Baoding 071002, P.R. China}
\email{wuyg1022@sina.com}

\date{}



\maketitle

\begin{abstract} 3-Lie algebras  are constructed  by Lie algebras, derivations and linear
functions, associative commutative  algebras, whose involutions and
derivations. Then the $3$-Lie algebras are obtained from group
algebras $F[G]$. An infinite dimensional simple $3$-Lie algebra $(A,
[ , , ]_{\omega, \delta_0})$ and a non-simple $3$-Lie algebra  $(A,
[ , , ]_{\omega_1, \delta})$ are constructed by Laurent polynomials
$A=F[t, t^{-1}]$ and its involutions $\omega$ and  $\omega_1$ and
derivations $\delta$ and $\delta_0$. At last of the paper, we
summarize the methods of constructing $n$-Lie algebras for $n\geq 3$
and provide a problem.
\end{abstract}

\baselineskip=18pt

\section{Introduction}
\setcounter{equation}{0}
\renewcommand{\theequation}
{1.\arabic{equation}}

To construct the generalized Hamiltonian dynamics, Nambu \cite{N, T}
first proposed the notion of 3-bracket. Thus Nambu dynamics is
described by the phase flow given by Nambu-Hamilton equations of
motion which involves two Hamiltonians. The notion of $n$-Lie
algebra was introduced by Filippov in 1985 (\cite{F}). It is a
natural generalization of the concept of a Lie algebra to the case
where the fundamental multiplication is n-ary, $n\geq 2$ (when $n =
2$ the definition agrees with the usual definition of a Lie
algebra). $3$-Lie algebras have close relationships with many
important fields in mathematics and mathematical physics (cf.
\cite{N, T, BL,HHM,HCK,G,P}). For example, the metric $3$-Lie
algebras are used to describe a world volume of multiple $M2$-branes
 \cite{BL, G} (BLG).

 Since the multiple multiplication, the
structure of $n$-Lie algebras is more complicated than that of Lie
algebras. Especially, the realizations of $n$-Lie algebras is
hardness.

Filippov in \cite{F2} constructed $n$-Lie algebra structure  on
commutative associative algebras by commuting derivations
 $\{D_1, \cdots, D_n\}$.  Let $A$ be a commutative associative
 algebra, then $A$ is an $n$-Lie algebra with the multiplication:
\begin{equation}[x_{1}, \cdots,
x_n]=\det \left(
                                        \begin{array}{ccc}
                                          D_1(x_1) &\cdots &D_1(x_n)\\
                                          \cdots &\cdots &\cdots\\
                                          D_n(x_1) &\cdots &D_n(x_n)\\
                                        \end{array}
                                      \right),\end{equation} which is called a
{\it Jcobian algebra} defined by $\{D_1, \cdots, D_n\}$.

Filippov and Pozhidaev in papers \cite{F2, APP, APPS} defined {\it
monomial $n$-Lie algebra} $A_G(f, t)$ with the multiplication
\begin{equation}[e_{a_1}, \cdots, e_{a_n}]=f(a_1, \cdots,
a_n)e_{a_1+\cdots+a_n+t}, \end{equation} where $G$ is an Abelian
group, $A_G$ is a vector space with a basis $\{ e_g | g\in G\}$,
$t\in G$, $f: G^{n}\rightarrow F$.

 Dzhumadildaev in  \cite{D1} constructed an $(n+1)$-Lie algebra from  a strong
$n$-Lie-Poisson algebra  $(L, \cdot, \omega)$ by endowing  a
skew-symmetric $(n+1)$-multiplication
\begin{equation}\bar{\omega}= Id_L\wedge
\omega, \end{equation} where $Id_L$ is the identity map of $L$.

Bai and collaborators in \cite{BW2, BBW} constructed $n$-Lie
algebras by two-dimensional extensions of metric n-Lie algebras. And
obtain $(n+1)$-Lie algebras by  $n$-Lie algebras and linear
functions $f\in L^*$, which satisfies $f(L^1)=0.$

In this paper, we pay our main attention to construct 3-Lie
algebras. The paper is organized as follows. Section $2$ introduces
some basic notions of $n$-Lie algebras. Section $3$ provides a
construction of 3-Lie algebras from a commutative associative
algebras, involutions and derivations. Section $4$ constructs
$3$-Lie algebras from group algebras.  Section $5$ studies $3$-Lie
algebras constructed by
 Laurent polynomials.

\section{\textbf{Fundamental notions of $n$-Lie algebra}}
\setcounter{equation}{0}
\renewcommand{\theequation}
{2.\arabic{equation}}

An $n$-Lie algebra $L$ (cf.\cite{ F}) is a vector space endowed with
an $n$-ary multilinear  skew-symmetric multiplication satisfying the
$n$-Jacobi identity: $\forall x_1, \cdots, x_n, y_2, \cdots, y_n\in
L$
\begin{equation}
  [[x_1, \cdots, x_n], y_2, \cdots, y_n]=\sum_{i=1}^n[x_1, \cdots, [ x_i, y_2, \cdots, y_n], \cdots,
  x_n].
  \end{equation}
The $n$-ary skew-symmetry of the
operation $[x_1, \cdots, x_n]$ means that
\begin{equation}
 [x_1, \cdots, x_n]=sgn(\sigma)[x_{\sigma (1)}, \cdots,
 x_{\sigma(n)}],\;\;\forall x_1, \cdots, x_n\in L,
 \end{equation}
for any permutation $\sigma\in S_n$.

A subspace $B$ of $L$ is called a {\it subalgebra} if
$[B,\dots,B]\subseteq B$. In particular, the subalgebra generated by
the vectors $[x_1, \cdots, x_n]$ for any $x_1, \cdots, x_n\in L$ is
called the {\it derived algebra} of $L$, which is denoted by $L^1$.
If $L^1=0$, then $L$ is called an  {\it abelian $n$-Lie algebra}.

{\it A  derivation} of an $n$-Lie algebra is a linear transformation
$D$ of $L$ into itself satisfying
\begin{equation}
    D([x_1, \cdots, x_n])=\sum_{i=1}^n[x_1, \cdots, D(x_i), \cdots, x_n], \label{E:2.4}
\end{equation}
for $x_1$, $\cdots$, $x_n \in L$. Let $Der L$ be the set of all
derivations of $L$. Then $Der L$ is a subalgebra of the general Lie
algebra $gl(L)$ and is called {\it the  derivation algebra} of $L$.

The map ad $(x_1, \cdots, x_{n-1}): L\rightarrow L$ given by
$$
    ad(x_1,\cdots, x_{n-1})(x_n) = [x_1, \cdots, x_{n-1}, x_{n}], \forall  x_n\in L,
$$

\vspace{2mm}\noindent is referred to as {\it a left multiplication}
defined by elements $x_1$, $\cdots$, $x_{n-1} \in L.$ It follows
from identity (2.1) that $ad(x_1, \cdots, x_{n-1})$ is a derivation.
The set of all finite linear combinations of left multiplications is
an ideal of $Der L$, which is denoted by $ad(L)$. Every derivation
in $ad(L)$ is by definition an inner derivation.

An {\it ideal} of an $n$-Lie algebra $L$ is a subspace $I$ such that
$ [I, L, \cdots, L]\subseteq I.$  If $L^1\neq 0$ and $L$ has no
ideals except for $0$ and itself, then $L$ is called a {\it simple
$n$-Lie algebra}.

\begin{flushleft}
\section{$3$-Lie algebras constructed by commutative associative algebras}
\end{flushleft}
\setcounter{equation}{0}
\renewcommand{\theequation}
{3.\arabic{equation}}

 R.Bai and collaborators (cf \cite{BBW}) constructed
$3$-Lie algebra $(L, [ , , ]_{\alpha})$ by a Lie algebra $(L, [ ,
])$ and linear function $\alpha\in L^*$, where $\alpha([L, L])=0$,
and $[ , , ]_{\alpha}$ defined as follows
 $$[a,b,c]_{\alpha}=\alpha(c)[a,b]+\alpha(a)[b,c]+
\alpha(b)[c,a], ~ for ~ a, b, c\in L.$$

 And  it is proved   that  the
$3$-Lie algebra $(L, [ , , ]_{\alpha})$ is a two step $3$-solvable
$3$-Lie algebra, that is, $L^{(2)}=[L^1, L^1, L^1]_{\alpha}=0$,
where $L^1=[L, L, L]_{\alpha}$.

In this section, we try to construct $3$-Lie algebras by commutative
associative algebras and their derivations, involutions and  linear
functions.

Let $A$ be a a commutative associative algebra over a field $F$. A
derivation $\Delta$ of $A$ is a linear mapping of $A$ satisfying
$\Delta(xy)=\Delta(x)y+x\Delta(y)$ for every $x, y\in A$. If a
linear mapping  $\omega: A\rightarrow A$ satisfying for every $a,
b\in A$, $\omega(ab)=\omega(a)\omega(b)$ and $\omega^2(a)=a$, then
$\omega$ is called {\it an involution} of $A$. If  $ch F\neq 2$,
then for every $c\in A$,

\vspace{2mm}
 $\omega(c+\omega(c)=c+\omega(c)$, ~~$\omega(c-\omega(c))=-(c-\omega(c))$, ~~ and $c=\frac{1}{2}(c+\omega(c))+\frac{1}{2}(c-\omega(c))$.

 \vspace{2mm}\noindent Therefore, $$A=A_1\dot+ A_{-1}, ~\mbox{where} ~A_1=\{ a|~ a\in A, \omega(a)=a~ \},~
A_{-1}=\{ b|~ b\in A, \omega(b)=-b~ \}.$$

 First we
give the following result.

\begin{lemma}
 Let $A$ be a  commutative associative
algebra, $\omega$ be an involution of $A$ and $\Delta\in Der A$.
Then $(A, [, ]_{\Delta})$ and  $(A, [, ]_{\omega})$  are Lie
algebras, where
\begin{equation}
[a, b]_{\Delta}=(Id_A\wedge \Delta)(a, b)= a\Delta(b)-b\Delta(a)~ a,
b\in A;
\end{equation}
\begin{equation}
[a, b]_{\omega}=(\omega\wedge Id_A)(a, b)=\omega(a)b-\omega(b)a, ~
a, b\in A.
\end{equation}

  If $\Delta$ and $\omega$ satisfy $\Delta \omega+\omega\Delta=0$,
then $(A, [, ]_{\omega, \Delta})$ is a Lie algebra, where
\begin{equation}
[a, b]_{\omega,\Delta}=((Id_A-\omega)\wedge \Delta)(a, b)=
(a-\omega(a))\Delta(b)-(b-\omega(b))\Delta(a), ~ a, b\in A.
\end{equation}

\end{lemma}

\begin{proof}  Eq. (3.1) is well know result. By the direct
computation,  Eqs. (3.2) and (3.3) satisfy  Eq (2.1), respectively.
\end{proof}

\begin{theorem}
 Let $A$ be a  commutative associative
algebra, $\Delta\in Der(A)$ and $\omega$ be an involution of $A$. If
$\alpha$, $\beta, \gamma\in A^*$ satisfy
\begin{equation}\alpha(ab)=0, ~~ \beta(a\Delta(b)-b\Delta(a))=0,
\gamma(a\Delta(b)-b\Delta(a))=\gamma(\omega(a)\Delta(b)-\omega(b)\Delta(a)),\end{equation}
then  $(A, [ , , ]_{\alpha,f})$ and $(A, [ , , ]_{\beta, \Delta})$
are $3$-Lie algebras, where  for arbitrary $a, b, c\in A,$
\begin{equation}
[a, b, c]_{\alpha, \omega}=(\alpha\wedge Id_A\wedge \omega)(a, b,
c)=\alpha(a)[b, c]_{\omega}+\alpha(b)[c, a]_{\omega}+\alpha(c)[a,
b]_{\omega},\end{equation}
\begin{equation} [a, b, c]_{\beta, \Delta}=(\beta\wedge
Id_A\wedge \Delta)(a, b, c)=\beta(a)[b, c]_{\Delta}+\beta(b)[c,
a]_{\Delta}+\beta(c)[a, b]_{\Delta}.\end{equation}

If $\Delta$ and $\omega$ satisfy $\Delta \omega+\omega\Delta=0$,
then $(A, [, ]_{\gamma, \omega, \Delta})$ is a $3$-Lie algebra,
where
\begin{equation}
[a, b, c]_{\gamma, \omega,\Delta}=(\gamma\wedge (Id_A-\omega)\wedge
\Delta)(a, b,
c)=\gamma(a)((b-\omega(b))\Delta(c)-(c-\omega(c))\Delta(b))\end{equation}
$$\hspace{8mm}+\gamma(b)((c-\omega(c))\Delta(a)-(a-\omega(a))\Delta(c))+\gamma(c)((a-\omega(a))\Delta(b)-(b-\omega(b))\Delta(a)).$$
\end{theorem}

{\begin{proof} Let  $\alpha\in A^*$. Then  for every $a, b\in A$,
\begin{center} $\alpha(\omega(a)b-\omega(b)a)=
\begin{cases}
0,&  \text{if} ~ a, b\in A_1;  ~~ \text{or } ~ a, b\in A_{-1},\\
2\alpha(ab),&  \text{if} ~ a \in A_1, b\in A_{-1},\\
-2\alpha(ab),&  \text{if} ~ b \in A_1, a\in A_{-1}.
\end{cases}$
\end{center} Therefore, if $\alpha(ab)=0$, then $\alpha(\omega(a)b-\omega(b)a)=0$.
From Lemma 3.1, Eq. (3.1),  and Theorem 3.1 in \cite{BBW}, the
result holds. \end{proof}

\begin{theorem} Let $A$ be a commutative associative algebra  over a field $F$, $\Delta$ be a derivation of $A$ and
$\omega: A\rightarrow A$ be an involution of $A$ satisfying
$$f\Delta+\Delta \omega=0.$$ Then $A$ is a $3$-Lie algebra in the
multiplication $[ , , ]_{\omega, \Delta}: A\otimes A\otimes
A\rightarrow A$, $ ~ \forall ~ a,b,c\in A,$

\begin{equation}
\hspace{-2cm}[a, b, c]_{\omega, \Delta}=\omega\wedge
Id_A\wedge\Delta(a, b, c)=\begin{vmatrix}
\omega(a) & \omega(b) &\omega(c) \\
a & b & c  \\
\Delta(a) & \Delta(b) & \Delta(c) \\
\end{vmatrix}
\end{equation}

\vspace{2mm}\hspace{3.6cm}$=\omega(a)(b\Delta(c)-c\Delta(b))+\omega(b)(c\Delta(a)-a\Delta(c))+
\omega(c)(a\Delta(b)-b\Delta(a)).$

\end{theorem}

\begin{proof} It is clear that $[, ,]_{\omega,\Delta}$ is a
$3$-ary linear skew-symmetric multiplication on $A$. Now we prove
that $[, , ]_{\omega,\Delta}$ satisfies Eq.(2.1). Since $A$ is
commutative and $\omega$ is an involution of $A$ which satisfies
$\omega\Delta+\Delta \omega=0$, by Eq.(3.8),  $\forall a,b,c,d,e \in
A$

\vspace{2mm}$\omega([a, b, c]_{\omega, \Delta})=
\begin{vmatrix}
\omega(a) & \omega(b) &\omega(c) \\
a & b & c  \\
\Delta \omega(a) & \Delta \omega(b) & \Delta \omega(c) \\
\end{vmatrix},$

\vspace{4mm}$\Delta([a, b, c]_{\omega, \Delta})=
\begin{vmatrix}
\Delta \omega(a) & \Delta \omega(b) & \Delta \omega(c) \\
a & b & c  \\
\Delta (a) & \Delta (b) & \Delta (c) \\
\end{vmatrix}
+\begin{vmatrix}
\omega(a) & \omega(b) &\omega(c) \\
a & b & c  \\
\Delta^2 (a) & \Delta^2 (b) & \Delta^2 (c) \\
\end{vmatrix},$

\vspace{4mm}$[[a, b, c]_{\omega,\Delta}, d, e]_{\omega,\Delta}=
\begin{vmatrix}
\omega([a, b, c]_{\omega, \Delta}) & \omega(d) &\omega(e) \\
[a, b, c]_{\omega, \Delta}    & d & e  \\
\Delta([a, b, c]_{\omega, \Delta}) & \Delta (d) & \Delta (e) \\
\end{vmatrix}$

\vspace{4mm} $=\begin{vmatrix}
\omega(a) & \omega(b) &\omega(c) \\
a & b & c  \\
\Delta \omega(a) & \Delta \omega(b) & \Delta \omega(c) \\
\end{vmatrix}
\begin{vmatrix}
d &e \\
\Delta (b) & \Delta (c) \\
\end{vmatrix}
-\begin{vmatrix}
\omega(a) & \omega(b) &\omega(c) \\
a & b & c  \\
\Delta(a) & \Delta(b) & \Delta(c) \\
\end{vmatrix}
\begin{vmatrix}
\omega(d) & \omega(e) \\
\Delta (d) & \Delta (e) \\
\end{vmatrix}$

\vspace{4mm} $+\begin{vmatrix}
\Delta \omega(a) & \Delta \omega(b) & \Delta \omega(c) \\
a & b & c  \\
\Delta (a) & \Delta (b) & \Delta (c) \\
\end{vmatrix}
\begin{vmatrix}
\omega(d) & \omega(e) \\
d & e \\
\end{vmatrix}
+\begin{vmatrix}
\omega(a) & \omega(b) &\omega(c) \\
a & b & c  \\
\Delta^2 (a) & \Delta^2 (b) & \Delta^2 (c) \\
\end{vmatrix}
\begin{vmatrix}
\omega(d) & \omega(e) \\
d & e \\
\end{vmatrix}.$

\vspace{2mm}\noindent Then we have

\vspace{2mm}$[[a,d,e]_{\omega,\Delta},b,c]_{\omega,\Delta}+[[b,d,e]_{\omega,\Delta},c,a]_{\omega,\Delta}+[[c,d,e]_{\omega,\Delta},a,b]_{\omega,\Delta}$

\vspace{4mm} $=\circlearrowleft_{a,b,c}\Delta
\omega(a)\begin{vmatrix}
b & c \\
\Delta (b) & \Delta (c) \\
\end{vmatrix}
\begin{vmatrix}
\omega(d) & \omega(e) \\
d & e \\
\end{vmatrix}
+\circlearrowleft_{a,b,c}a\begin{vmatrix}
\omega(b) & \omega(c) \\
\Delta (b) & \Delta (c) \\
\end{vmatrix}
\begin{vmatrix}
\omega(d) & \omega(e) \\
\Delta (d) & \Delta (e) \\
\end{vmatrix}$

\vspace{4mm} $+\circlearrowleft_{a,b,c}\Delta
\omega(a)\begin{vmatrix}
\omega(b) & \omega(c) \\
b & c \\
\end{vmatrix}
\begin{vmatrix}
d & e \\
\Delta (d) & \Delta (e) \\
\end{vmatrix}
+\circlearrowleft_{a,b,c}\Delta^2 (a)\begin{vmatrix}
\omega(b) & \omega(c) \\
b & c \\
\end{vmatrix}
\begin{vmatrix}
\omega(d) & \omega(e) \\
d & e \\
\end{vmatrix}$

\vspace{4mm}$ =[[a,b,c]_{\omega,\Delta},d,e]_{\omega,\Delta},$
\\where $\circlearrowleft_{a,b,c}$ is the circulation of $a, b, c$,
for example

\vspace{2mm}\noindent$\circlearrowleft_{a,b,c}\Delta
\omega(a)\begin{vmatrix}
b & c \\
\Delta (b) & \Delta (c) \\
\end{vmatrix}=\Delta
\omega(a)\begin{vmatrix}
b & c \\
\Delta (b) & \Delta (c) \\
\end{vmatrix}+\Delta
\omega(b)\begin{vmatrix}
c & a \\
\Delta (c) & \Delta (a) \\
\end{vmatrix}+\Delta
\omega(c)\begin{vmatrix}
a & b \\
\Delta (a) & \Delta (b) \\
\end{vmatrix}.$

\vspace{2mm}Therefore, $(A, [ , , ]_{\omega,\Delta})$ is a $3$-Lie
algebra in the multiplication (3.8). \end{proof}

\begin{coro}  Let $A$ be a
commutative associative algebra, $\Delta\in Der(A)$ and $\omega$ be
an involution of $A$ satisfying $\Delta\omega+\omega\Delta=0$. Then
 $3$-Lie algebra  $(A, [ , , ]_{ \Delta, \omega})$
 with the multiplication (3.8) can be decomposed into the direct sum of abelian subalgebras $A_{1}$ and $A_{-1}$, that
 is, $A=A_1\dot+ A_{-1},$ and
$$ [A_1, A_1,
A_1]_{\omega,\Delta}=[A_{-1}, A_{-1}, A_{-1}]_{\omega,\Delta}=0. $$
 And $\Delta(A_1)\subseteq A_{-1}, ~ \Delta (A_{-1})\subseteq
A_1.$

\end{coro}

\begin{proof} For every $a_i\in A_1,$
 $\omega(a_i)=a_i$ for $i=1, 2, 3$, $\omega\Delta(a_1)=-\Delta \omega(a_1)=-\Delta(a_1)$, then

\vspace{2mm}
  $  [a_1, a_2,
    a_3]_{\omega, \Delta}=\begin{vmatrix}
\omega(a_1)& \omega(a_2)&
    \omega(a_3) \\
a_1& a_2&
    a_3  \\
\Delta(a_1)& \Delta(a_2)&
   \Delta(a_3) \\
\end{vmatrix}=\begin{vmatrix}
a_1& a_2&
    a_3 \\
a_1& a_2&
    a_3  \\
\Delta(a_1)& \Delta(a_2)&
   \Delta(a_3) \\
\end{vmatrix}=0$.

\vspace{2mm}\noindent Thus $A_1$ is an abelian subalgebra, and
$\Delta(A_1)\subseteq A_{-1}$. \end{proof}

\begin{theorem} Let  $(A, [ , , ]_{\omega, \Delta})$ be the $3$-Lie algebra with the multiplication (3.8). If  $I$ is
an ideal of the  associative algebra $A$ which satisfies
$\omega(I)\subseteq I$ and $\Delta(I)\subseteq I$. Then $I$ is an
ideal of the $3$-Lie algebra $(A, ~ [ , , ]_{\omega, \Delta})$.
\end{theorem}

\begin{proof}  Let $I$ be an ideal of the associative
algebra $A$ satisfying $\omega(I)\subseteq I$ and
$\Delta(I)\subseteq I$. Then for every  $a\in I$,   and   $\forall
b, c\in A$, by Eq. (3.8)
 $$
[a, b, c]_{\omega,\Delta}=\omega(a)(b\Delta(c)-c\Delta(b))+
\omega(b)(a\Delta(c)-c\Delta(a))+\omega(c)(a\Delta(b)-b\Delta(a))\in
I.$$

 \vspace{2mm}\noindent Therefore, $I$ is an ideal of the $3$-Lie algebra $(A, ~ [ , , ]_{\omega,\Delta})$.
\end{proof}

\begin{theorem} Let $A_1$ and $A_2$ be commutative associative algebras, $\omega_i$ be an involution of $A_i$ and $\Delta_i$ be a derivation of $A_i$ which satisfies $\omega_i\Delta_i+\Delta_i
\omega_i=0$ for $i=1, 2$. If $\sigma : A_1\rightarrow A_2$ is an
associative algebra isomorphism satisfying $\sigma
\omega_1=\omega_2\sigma,$ and $\sigma\Delta_1=\Delta_2\sigma.$ Then
$\sigma$ is a $3$-Lie algebra isomorphism from $(A_1, [ , ,
]_{\omega_1, \Delta_1}) $  onto $(A_2, [ , , ]_{\omega_2, \Delta_2})
$.

\end{theorem}

\begin{proof}  For every $a, b, c\in A_1$, by Eq.(3.8),

\vspace{2mm}\noindent $\sigma([a, b, c]_{\omega_1,
\Delta_1})=\sigma(\omega_1(a)(b\Delta_1(c)-c\Delta_1(b))+\omega_1(b)(c\Delta_1(a)-a\Delta_1(c))+
\omega_1(c)(a\Delta_1(b)-b\Delta_1(a)))$

\vspace{2mm}\noindent\hspace{2.4cm} $=\omega_2(\sigma(a))(\sigma
(b)\Delta_2\sigma(c)-\sigma(c)\Delta_2(\sigma(b)))+\omega_2(\sigma(b))(\sigma(c)\Delta_2(\sigma(a))-\sigma(a)\Delta_2(\sigma(c)))$

\vspace{2mm}\noindent\hspace{2.5cm}$+
\omega_2(\sigma(c))(\sigma(a)\Delta_2(\sigma(b))-\sigma(b)\Delta_2(\sigma(a)))$

\vspace{2mm}\noindent\hspace{2.5cm}$=[\sigma(a), \sigma(b),
\sigma(c)]_{\omega_2, \Delta_2}.$

It follows the result. \end{proof}

\begin{coro}  {\it  If $A$ is a
nilpotent commutative associative algebra, then $(A, [ , ,
]_{\omega, \Delta}) $ is a nilpotent $3$-Lie algebra.}

\end{coro}

\begin{proof} Suppose  $A^m=0$ for some positive integer
$m$. Then for every $a_1, \cdots, a_m\in A$, $a_1\cdots a_m=0$.
Thanks to
 Eq. (3.8), for every $a, b\in A$, $ad^m(a, b)(A)\subseteq A^m=0$. Therefore, $(A, [ , , ]_{\omega, \Delta}) $ is a nilpotent $3$-Lie
 algebra. \end{proof}

\begin{flushleft}
\section{$3$-Lie algebras constructed by group algebras }
\end{flushleft}
\setcounter{equation}{0}
\renewcommand{\theequation}
{4.\arabic{equation}}

 Let $(G, +)$ be an additive
Abelian group, and $F[G]$ be the group algebra, that is, $F[G]$ is a
commutative associative algebra with a basis $\{ e_g ~ | g\in G\}$,
and for every $  x=\sum\limits_{g\in G}\lambda_g e_g$,
$y=\sum\limits_{h\in G}\mu_h e_h\in F[G]$,
$$x+y=\sum\limits_{g\in G}(\lambda_g +\mu_g)e_g, ~~ xy=(\sum\limits_{g\in G}\lambda_g e_g)(\sum\limits_{h\in
G}\mu_h e_h)=\sum\limits_{g, h\in G}\lambda_g\mu_h e_{g+h}.$$

Define linear mapping $\omega : F[G] \rightarrow F[G]$,

 \begin{equation}\omega(x)=\omega(\sum\limits_{g\in
G}\lambda_g e_g)=\sum\limits_{g\in G}\lambda_g e_{-g},~~ \forall
x=\sum\limits_{g\in G}\lambda_g e_g\in F[G].
\end{equation}
Then $\omega$ is a linear isomorphism of $F[G]$ and for  $
x=\sum\limits_{g\in G}\lambda_g e_g$, $y=\sum\limits_{h\in G}\mu_h
e_h\in F[G],$ $$\omega^2(x)=\omega(\sum\limits_{g\in G}\lambda_g
e_{-g})=\sum\limits_{g\in G}\lambda_g e_g=x,$$

$$\omega(xy)=\omega(\sum\limits_{g\in G}\lambda_g e_g \sum\limits_{h\in G}\mu_h
e_h)=\omega(\sum\limits_{g, h\in G}\lambda_g\mu_h
e_{g+h})=\sum\limits_{g, h\in G}\lambda_g\mu_h e_{-g-h}$$

\vspace{2mm}\hspace{2.8cm}$=\sum\limits_{g\in G}\lambda_g
e_{-g}\sum\limits_{h G}\mu_h e_{-h} =\omega(\sum\limits_{g\in
G}\lambda_g e_g)\omega(\sum\limits_{h\in G}\mu_h
e_h)=\omega(x)\omega(y).$

Therefore, $\omega$ is an involution of the commutative associative
algebra $F[G]$.

Denote $F^{+}$ the addition group of $F$.  For  every $\alpha\in Hom
(G, F^+)$, then $\alpha$ satisfies
$\alpha(g+h)=\alpha(g)+\alpha(h)$, $\forall g, h\in G$. Define
linear mapping $\alpha^{*} : F[G]\rightarrow F[G]$ as follows:

\begin{equation}\alpha^{*}(\sum\limits_{g\in G}\lambda_g e_g)=\sum\limits_{g\in
G}\lambda_g \alpha(g) e_g, ~~ \forall \sum\limits_{g\in G}\lambda_g
e_g\in F[G].\end{equation}

\begin{lemma}

Let $(G, +)$ be an  abelian group,  $\omega: F[G]\rightarrow F[G]$
be defined as Eq.(4.1). Then for every  $\alpha\in Hom(G, F^{+})$,
$\alpha^*$ defined as Eq. (4.2) is a derivation of the algebra
$F[G]$, and satisfies $\omega\alpha^*+\alpha^*\omega=0.$

\end{lemma}

\begin{proof}  By Eqs. (4.1) and (4.2), for arbitrary $
x=\sum\limits_{g\in G}\lambda_g e_g$, $y=\sum\limits_{h\in G}\mu_h
e_h\in F[G],$

$$\hspace{-2cm}\alpha^{*}((\sum\limits_{g\in G}\lambda_g e_g)(\sum\limits_{h\in G}\mu_h
e_h))=\alpha^{*}(\sum\limits_{g, h\in G}\lambda_g\mu_h
e_{g+h})=\sum\limits_{g, h\in G}\lambda_g\mu_h
(\alpha(g)+\alpha(h))e_{g+h}$$

$$=\sum\limits_{g, h\in G}\lambda_g\mu_h \alpha(g)e_{g+h}+\sum\limits_{g, h\in G}\lambda_g\mu_h \alpha(h)e_{g+h}=\alpha^{*}(\sum\limits_{g\in G}\lambda_g e_g)(\sum\limits_{h\in G}\mu_h
e_h)+(\sum\limits_{g\in G}\lambda_g e_g)\alpha^{*}(\sum\limits_{h\in
G}\mu_h e_h),$$

$$\hspace{-5cm}(\omega\alpha^*+\alpha^*\omega)(\sum\limits_{g\in G}\lambda_g e_g)=\omega(\sum\limits_{g\in G}\lambda_g\alpha(g)
e_g)+\alpha^*(\sum\limits_{g\in G}\lambda_g e_{-g})
$$
$$\hspace{-3cm}=\sum\limits_{g\in
G}\lambda_g\alpha(g) e_{-g}+\sum\limits_{g\in G}\lambda_g
\alpha(g^{-1})e_{-g}=\sum\limits_{g\in G}\lambda_g\alpha(g)
e_{-g}-\sum\limits_{g\in G}\lambda_g \alpha(g)e_{-g}=0.$$ It follows
the result. \end{proof}

\begin{theorem}
Let $G$ be an  abelian group,  $\omega: F[G]\rightarrow F[G]$ be
defined as Eq.(4.1), $\alpha\in Hom(G, F^{+})$, $\alpha^*$ be
defined as Eq. (4.2). Then $(F[G], [ , , ]_{\omega,\alpha^*})$ is a
$3$-Lie algebra, where  for arbitrary  $ \sum\limits_{g\in
G}\lambda_g e_g,$ $\sum\limits_{h\in G}\mu_h e_h,$
$\sum\limits_{q\in G}\nu_q e_q\in F[G],$

\begin{equation} [\sum\limits_{g\in G}\lambda_g
e_g,\sum\limits_{h\in G}\mu_h e_h,\sum\limits_{q\in G}\nu_q
e_q]_{\omega,\alpha^*}=\sum\limits_{g, h, q\in G}\lambda_g\mu_h\nu_q
( \alpha(q-h)e_{h+q-g}\end{equation}

\vspace{2mm}\hspace{7.5cm}$+\alpha(g-q)e_{g+q-h}+
\alpha(h-g)e_{g+h-q}).$

\end{theorem}

\begin{proof}  By Lemma 4.1 and Theorem 3.3, for arbitrary  $
\sum\limits_{g\in G}\lambda_g e_g,$ $\sum\limits_{h\in G}\mu_h e_h,$
$\sum\limits_{q\in G}\nu_q e_q\in F[G],$

$$[\sum\limits_{g\in G}\lambda_g e_g,\sum\limits_{h\in G}\mu_h
e_h,\sum\limits_{q\in G}\nu_q e_q]_{\omega,\alpha^*}
 =\begin{vmatrix} \sum\limits_{g\in G}\lambda_g e_{-g} & \sum\limits_{h\in G}\mu_h e_{-h} & \sum\limits_{q\in G}\nu_q e_{-q} \\
\sum\limits_{g\in G}\lambda_g e_g & \sum\limits_{h\in G}\mu_h e_h & \sum\limits_{q\in G}\nu_q e_q  \\
\sum\limits_{g\in G}\lambda_g\alpha(g) e_g & \sum\limits_{h\in G}\mu_h\alpha(h) e_h & \sum\limits_{q\in G}\nu_q\alpha(q) e_q  \\
\end{vmatrix}$$

\vspace{2mm}\hspace{1.2cm}$=\sum\limits_{g, h, q\in
G}\lambda_g\mu_h\nu_q \{
(\alpha(q)-\alpha(h))e_{h+q-g}+(\alpha(g)-\alpha(q))e_{g+q-h}+(\alpha(h)-\alpha(g))e_{g+h-q}\}$

\vspace{2mm}\hspace{1.2cm}$=\sum\limits_{g, h, q\in
G}\lambda_g\mu_h\nu_q ( \alpha(q-h)e_{h+q-g}+\alpha(g-q)e_{g+q-h}+
\alpha(h-g)e_{g+h-q}).$

It follows the result. \end{proof}

By the above discussions, the products  of  basis vectors $\{ e_g ~
|~ g\in G\} $  of $3$-Lie algebra $(F[G], [ , , ]_{\omega,
\alpha^*})$ are as follows: for arbitrary $g, h, w\in G$,

\begin{equation}\hspace{-3cm}[e_g, e_h, e_w]_{\omega, \alpha^*}=\begin{vmatrix} e_{-g}&  e_{-h}&
     e_{-w} \\
e_g& e_h&
   e_w  \\
 \alpha(g)e_g& \alpha(h)e_h&
  \alpha(w)e_w \\
\end{vmatrix}\end{equation}

\vspace{2mm}\hspace{4.8cm}$=
\alpha(w-h)e_{h+w-g}+\alpha(g-w)e_{g+w-h}+ \alpha(h-g)e_{g+h-w}. $

For $\alpha\in Hom (G, F^+)$, define  mapping

\begin{equation} \phi_{\alpha}: F[G] \rightarrow F, ~~  \phi_{\alpha}(x)=\sum\limits_{g\in
G}\lambda_g\alpha(g), ~~ \forall~~ x=\sum\limits_{g\in G}\lambda_g
e_g\in F[G].
\end{equation}

 Denote ~ $I_0=\{x |~ x=\sum\limits_{g\in
G}\lambda_g e_g\in F[G], ~ \phi_{\alpha}(x)=\sum\limits_{g\in
G}\lambda_g \alpha(g)=0\}.$ Then $I_0$ is a subspace of $F[G]$. And
we have the following result.

\begin{theorem}

For $\alpha\in Hom (G, F^+)$, if $\alpha\neq 0$, then $I_0$ is a
maximal  ideal of $3$-Lie algebra $(F[G], [ , ,
]_{\omega,\alpha^*})$. Therefore, $(F[G], [ , ,
]_{\omega,\alpha^*})$ is a non-simple $3$-Lie algebra.
\end{theorem}

\begin{proof}  For arbitrary  $ g, h, w\in G$, by Eq. (4.4),

$\phi_{\alpha}([e_g, e_h, e_w]_{\omega,
\alpha^*})=\phi_{\alpha}(\alpha(w-h)e_{h+w-g}+\alpha(g-w)e_{g+w-h}+
\alpha(h-g)e_{g+h-w})$

\vspace{2mm}\hspace{2.8cm}$=(\alpha(w)-\alpha(h))(\alpha(h)+\alpha(w)-\alpha(g))$+$(\alpha(g)-\alpha(w))(\alpha(g)+\alpha(w)-\alpha(h))$

\vspace{2mm}\hspace{2.8cm}+$(\alpha(h)-\alpha(g))(\alpha(g)+\alpha(h)-\alpha(w))$$=0$.
\\
It follows that the derived algebra of $(F[G], [ , ,
]_{\omega,\alpha^*})$ is contained in $I_{0}$. Therefore,  $I_0$ is
an ideal of the $3$-Lie algebra $(F[G], [ , , ]_{\omega,\alpha^*})$.

Since $\alpha\neq 0$, without loss of generality suppose
$\alpha(d)=1$ for some non-zero element $d$ of $G$. Then for every
$x\in F[G]$, $x=\phi_{\alpha}(x)e_d+(x-\phi_{\alpha}(x)e_d)$. Since
$$\phi_{\alpha}(x-\phi_{\alpha}(x)e_d)=\phi_{\alpha}(x)-\phi_{\alpha}(x)=0,$$
we have  $F[G]=Fe_d\dot+ I_0$ as the direct sum of subspaces.
Therefore, $I_0$ is a maximal  ideal of $3$-Lie algebra $(F[G], [ ,
, ]_{\omega,\alpha^*})$. ~~\end{proof}

 \vspace{2mm}\noindent{\bf Example 4.1 } Let $G=\{ A=(a_{ij})| a_{ij}\in F, 1\leq i\leq m, 1\leq j\leq n\}
 $ be the set of all $(m\times n)$-matrices over a field $F$. Then
 $G$ is an abelian group in the addition:   $\forall A=(a_{ij}), B=(b_{ij})\in G$,
 $A+B=(a_{ij}+b_{ij})$.
Define $\alpha: G\rightarrow F^+$ and $\omega: F[G]\rightarrow F[G]$
as follows
 $$\alpha(A)=\frac{1}{2}\sum\limits_{i=1}^m\sum\limits_{j=1}^na_{ij},  ~~  \hspace{5mm}\omega(e_A)=e_{-A}.$$
 Then $\alpha\in Hom(G, F^+)$ and $\omega$ is an involution of $F[G]$. By Theorem 4.2,
   $F[G]$ is an $mn$-dimensional $3$-Lie algebra with the multiplication:  $\forall A=(a_{ij}), B=(b_{ij}), C=(c_{ij})\in G,$
$$
 [e_A, e_B,
 e_C]_{\omega,\alpha^*}=\sum\limits_{i=1}^m\sum\limits_{j=1}^n(c_{ij}-b_{ij})e_{B+C-A}+\sum\limits_{i=1}^m\sum\limits_{j=1}^n(a_{ij}-c_{ij})e_{C+A-B}+\sum\limits_{i=1}^m\sum\limits_{j=1}^n(b_{ij}-a_{ij})e_{A+B-C}.
$$

\vspace{2mm}\noindent{\bf Example 4.2 } Let $G=\{ A=(a_{ij})|
a_{ij}\in F, 1\leq i, j\leq n\}
 $ be the group of  all $(n\times n)$-matrices over a field $F$ with the addition: $\forall A=(a_{ij}), B=(b_{ij})\in G$,
 $A+B=(a_{ij}+b_{ij})$. Let $\beta\in Hom (G, F^+)$,
 $\beta(A)=tr(A)=\sum\limits_{i=1}^na_{ii}$, $\omega: F[G]\rightarrow
 F[G]$, $\omega(e_A)=e_{-A}$. Then  by Theorem 4.2,  $F[G]$ is an $n^2$-dimensional $3$-Lie algebra in the multiplication:  $\forall A=(a_{ij}), B=(b_{ij}), C=(c_{ij})\in G,$
$$ [e_A, e_B,
 e_C]_{\omega,\beta^*}=tr(C-B)e_{B+C-A}+tr(A-C)e_{A+C-B}+tr(B-A)e_{A+B-C}.$$

 \vspace{2mm}\noindent{\bf Example 4.3}  Let $G =Z_p^+$  be the
 addition
group of the prime field $Z_p$, $ch F_p=p$. Then the multiplication
of the group algebra $Z_p[G]$ is ~~
$$e_{\bar{r}}e_{\bar s} =
e_{\overline{r+s}}, ~~\forall \bar s, \bar r\in G.$$ Define
$$\alpha : G\rightarrow Z_p^+ , ~~ \alpha(\bar r)=\bar r, ~~ \forall \bar r\in Z_p;$$
$$\omega: Z_p[G] \rightarrow Z_p[G],  ~~ \omega(e_{\bar r})=e_{\overline{-r}}, ~ ~ \bar r\in Z_p.$$

 By Theorem 4.2 ,  $(Z_p[G], [ , , ]_{\omega,\alpha^*})$ is
a $p$- dimensional $3$-Lie algebra with  the multiplication as
follows
$$
[e_{\bar r}, e_{\bar s}, e_{\bar k}]=\begin{vmatrix}
e_{\overline{-r}} & e_{\overline{-s}} & e_{\overline{-k}} \\
e_{\bar r} & e_{\bar s} & e_{\bar t}  \\
\bar r e_{\bar r} & \bar s e_{\bar s} & \bar k e_{\bar k} \\
\end{vmatrix}=\overline{k-s} e_{\overline{s+k-r}}+\overline{r-k} e_{\overline{k+r-s}}+ \overline{s-r} e_{\overline{r+s-k}}.
$$

\begin{flushleft}
\section{ $3$-Lie algebras  constructed from  Laurent polynomials }
\end{flushleft}
\setcounter{equation}{0}
\renewcommand{\theequation}
{5.\arabic{equation}}

In this section we study  $3$-Lie algebras constructed by Laurent
polynomials. In the following, denote $A=F[t^{-1}, t]$, the set of
Laurent polynomials over a field $F$.

We know that the derivation algebra $Der A=\{ t^s\frac{d}{dt}| ~
s\in Z\}$ with the product
\begin{equation}
[t^{m}\delta, t^{n}\delta]=(n-m)t^{m+n}\delta,~~ ~ m, n\in Z,
\end{equation}
where $\delta=t\frac{d}{dt}, $ $Z$ is the set of all integer
numbers.

Let $\omega: A \rightarrow A$ be an algebra homomorphism and satisfy
$\omega^2=Id_{A}.$ Since $\omega(1)=1$ and
$\omega(t)\omega(t^{-1})=1$, we have $ \omega(t)=\lambda t^r$,
$\lambda\in F$ and $\lambda\neq 0$. Then
$\omega(t^m)=\lambda^mt^{rm}$,

\vspace{2mm}
$$t^m=\omega^2(t^m)=\omega(\omega(t^m))=\omega((\omega(t))^m)=\omega(\lambda^m
t^{rm})=\lambda^{m+rm}t^{r^2m}.$$

\noindent We obtain $r=-1, ~ \lambda\neq 0;$  or $r=1$,
$\lambda=\pm1$. Therefore, we get the following result.

\begin{lemma} Let  $\omega:
A \rightarrow A$ be a linear map.
 If  $ch F\neq 2$, then $\omega$ is an involution of $A$ if and only
if $\omega$ satisfies
\begin{equation}
 \omega(t^m)=\varepsilon^m t^m, \varepsilon=\pm1, ~ \forall m\in Z; ~ or  ~ \omega(t^m)=\lambda^mt^{-m},  ~ \lambda\in F, ~ \lambda\neq 0, ~ \forall ~ m\in Z.
\end{equation}

\end{lemma}

\begin{proof} The result follows  from the above
discussions. ~~\end{proof}

\begin{lemma}
 Let   $\omega$ be an involution of $A$, $\delta=t^l\frac{d}{dt}\in Der F[t^{-1}, t]$ and
$ch F\neq 2$. Then $\omega\delta+\delta \omega=0$ if and only if
$\omega, \delta$ satisfy the following one possibilities

$ (i)\hspace{4mm} \omega(t^m)=(-1)^m t^m, ~ \forall m\in Z, ~~~
\delta=t^{2k}\frac{d}{dt}, k\in Z.$

$(ii) \hspace{4mm} \omega(t^m)=\lambda^mt^{-m},  ~ \lambda\in F, ~
\lambda\neq 0, ~ \forall ~ m\in Z, ~~~\delta=t\frac{d}{dt}.$

$(ii)$ \hspace{4mm} If $chF=2, ~~ \omega(t^m)=t^{-m},   ~ \forall ~
m\in Z, ~~~\delta=t\frac{d}{dt}.$

\end{lemma}

\begin{proof} The result follows  from Lemma 5.1
and the direct computation. ~~\end{proof}

\begin{theorem}

Let   $\delta=t\frac{d }{dt}$,  $ch F\neq 2$, $\omega_{\lambda}:
A\rightarrow A,$
$$ \omega_{\lambda}(t^m)=\lambda^m
t^{-m}, ~ \lambda\in F, ~ \lambda\neq 0, ~ \forall m\in Z.$$ Then
$(A, ~ [ , , ]_{\omega_{\lambda},\delta})$ is a $3$-Lie algebra in
the multiplication:  $\forall ~~t^l, t^m, t^n\in L$,

\begin{equation}\hspace{-5cm}[t^l, t^m,
t^n]_{\omega_{\lambda}, \delta}=\begin{vmatrix}
\lambda^lt^{-l} & \lambda^mt^{-m} &\lambda^nt^{-n} \\
t^l & t^m & t^n  \\
lt^l & mt^m & nt^n \\
\end{vmatrix}
\end{equation}

\vspace{2mm}\hspace{4cm}$=\lambda^l(n-m)t^{m+n-l}+\lambda^m(l-n)t^{n+l-m}+\lambda^n(m-l)t^{l+m-n}.$

\end{theorem}

\begin{proof} The result follows from Theorem
3.3 and Lemma 5.2. \end{proof}

\begin{coro} Let $\lambda=1$ in Theorem 5.3. Then  the multiplication of the $3$-Lie algebra $(A, [, , ]_{\omega_1, \delta})$ is as follows: $\forall~ t^l, t^m,
t^n\in A,$
\begin{equation}
  \hspace{-4cm}  [t^{l}, t^{m},
    t^{n}]_{\omega_1, \delta}=\begin{vmatrix}
t^{-l} & t^{-m} &t^{-n} \\
t^l & t^m & t^n  \\
lt^l & mt^m & nt^n \\
\end{vmatrix}
\end{equation}

\vspace{2mm}\hspace{5cm}$=(n-m)t^{m+n-l}+(l-n)t^{n+l-m}+(m-l)t^{l+m-n}.$

\vspace{2mm}And \hspace{2cm}$\omega([ t^l, t^m, t^n]_{\omega_1,
\delta})=-[ \omega(t^l), \omega(t^m), \omega(t^n)]_{\omega_1,
\delta}.$ ~~

\end{coro}

\begin{theorem}

Let $ch F\neq 2$. Then the $3$-Lie algebra $(A, [ , ,
]_{\omega_{\lambda}, \delta}) $ with the multiplication (5.3) for
some $\lambda\neq 0$ is isomorphic to the $3$-Lie algebra $(A, [ , ,
]_{\omega_1, \delta}) $ with the multiplication (5.4), where
$\delta=t\frac{d}{dt}$.
\end{theorem}

\begin{proof} Denote  $\omega_{\lambda}:
A\rightarrow A$, $\omega_{\lambda}(t^m)=\lambda^mt^{-m}, \lambda\in
F, \lambda\neq 0$. Then we have $\omega_{\lambda}\delta+\delta
\omega_{\lambda}=0.$

Define $\sigma: A\rightarrow A,$
$\sigma(t^m)=\lambda^{\frac{m}{2}}t^m$, $\forall ~ m\in Z$. Then

$$\sigma(t^m t^n)=\sigma(t^{m}) \sigma(t^{n}), ~~  \delta\sigma
(t^m)=\delta(\lambda^{\frac{m}{2}}t^{m})=m\lambda^{\frac{m}{2}}t^{m}=\sigma\delta(t^{m}),$$
$$\omega_{1}\sigma(t^m)=\omega_{1}(\lambda^{\frac{m}{2}}t^{m})=\lambda^{\frac{m}{2}}t^{-m}=\lambda^{\frac{-m}{2}}(\lambda^mt^{-m})=\sigma
\omega_{\lambda}(t^m).$$

\vspace{2mm}$\delta\sigma
(t^m)=\delta(\lambda^{\frac{m}{2}}t^{m})=m\lambda^{\frac{m}{2}}t^{m}=\sigma\delta(t^{m}).$

\vspace{2mm}\noindent Follows from Theorem 3.7, the result holds. ~~
\end{proof}

\begin{theorem} If $ch F=p > 2$, then for every integer $k\in Z$ and $k\neq 0$,
$$I_k=\{ ~(t^{kp}+t^{-kp})h(t)|~~ \forall h(t)\in A ~\},
~~~ J_k=\{ ~(t^{kp}-t^{-kp})h(t)|~~ \forall h(t)\in A ~\}$$

\vspace{2mm}\noindent are non-zero proper ideals of the $3$-Lie
algebra $(A, [ , , ]_{\omega_1, \delta}) $.
\end{theorem}

\begin{proof} Since for every integer $k\in Z$
and $k\neq 0$, $I_k, J_k$ are ideals of the associative algebra
$A=F[ t^{-1}, t ]$, and satisfy
$$\omega(I_k)\subseteq I_k,~~\omega(J_k)\subseteq J_k, ~~ \Delta(I_k)\subseteq I_k,~~ \Delta(J_k)\subseteq J_k. $$
By Theorem 3.5, $I_k$ and $J_k$ are proper ideals of the $3$-Lie
algebra. ~\end{proof}

By the above discussions, $A_1$ and $A_{-1}$ are two abelian
subalgebras of $(A, [ , , ]_{\omega_{1}, \delta}) $,  and
$A=A_1\dot+ A_{-1},$ where

$$ A_1=\{~ p(t)|~~ p(t)\in A, \omega(p(t))=p(t)~ \}=\{ ~
\sum\limits_{i=r}^sa_i(t^i+t^{-i}),~ a_i\in F, r, s\in Z~ \},$$

$$A_{-1}=\{~ p(t)|~~ p(t)\in A, \omega(p(t))=-p(t)~ \}=\{ ~
\sum\limits_{i=r}^sa_i(t^i-t^{-i}),~ a_i\in F, r, s\in Z ~ \}.$$

\vspace{2mm}If $B=F[t_1^{-1}, \cdots, t_k^{-1}, t_1, \cdots, t_k]$
is the commutative associative algebra of $k$ variable Laurent
polynomials over a field $F$ of characteristic zero. Then for every
$1\leq j\leq k$, $\delta_j=t_{j}\frac{\partial}{\partial t_j}$ are
derivations of $B$, where for every $p(t_1, \cdots,
t_k)=\sum\limits_{i_1\cdots i_k} a_{i_1\cdots i_k}t_1^{i_1}\cdots
t_{k}^{i_k}\in B,$ $\delta_j(p(t_1, \cdots,
t_k))=\sum\limits_{i_1\cdots i_k} i_ja_{i_1\cdots
i_k}t_1^{i_1}\cdots t_{k}^{i_k}$.

\begin{theorem}

Let  $B=F[t_1^{-1}, \cdots, t_k^{-1}, t_1, \cdots, t_k]$ with $ch
F=0$, $\Delta_j=t_j\frac{\partial}{\partial t_j}$ be a derivation of
$B$. For an algebra  homomorphism $\omega: B\rightarrow B$, $\omega$
satisfies $\Delta_j \omega+\omega\Delta_j=0$ and $\omega^2=id_B$ if
and only if
\begin{equation}
\omega(t_1^{r_1}\cdots t_j^{r_j}\cdots
t_k^{r_k})=(\lambda_1^{r_1}\cdots \lambda_j^{r_j}\cdots
\lambda_k^{r_k})t_1^{-r_1}\cdots t_j^{-r_j}\cdots t_k^{-r_k},
\end{equation}
\\ where $~ \lambda_s\in F, ~ \lambda_s\neq 0, ~
 ~ r_s\in Z, ~ 1\leq s\leq k.$
Therefore, $(B, ~ [ , , ]_{\omega, \delta_j})$ is a $3$-Lie algebra
in the multiplication:  $\forall ~~t_1^{r_1}\cdots t_j^{r_j}\cdots
t_k^{r_k}, ~ t_1^{i_1}\cdots t_j^{i_j}\cdots t_k^{i_k}, ~
t_1^{n_1}\cdots t_j^{n_j}\cdots t_k^{n_k}\in B$,

\begin{equation}[t_1^{r_1}\cdots t_j^{r_j}\cdots
t_k^{r_k}, ~~ t_1^{i_1}\cdots t_j^{i_j}\cdots t_k^{i_k}, ~~
t_1^{n_1}\cdots t_j^{n_j}\cdots t_k^{n_k}]_{\omega,
\Delta_j}\end{equation}

\hspace{2.7cm}$=(\lambda_1^{r_1}\cdots \lambda_j^{r_j}\cdots
\lambda_k^{r_k})(n_j-i_j)t_1^{i_1+n_1-r_1}\cdots
t_j^{i_j+n_j-r_j}\cdots t_k^{i_k+n_k-r_k}$

\hspace{2.7cm}$+(\lambda_1^{i_1}\cdots \lambda_j^{i_j}\cdots
\lambda_k^{i_k})(r_j-n_j)t_1^{r_1+n_1-i_1}\cdots
t_j^{r_j+n_j-i_j}\cdots t_k^{r_k+n_k-i_k}$

\hspace{2.7cm}$+(\lambda_1^{n_1}\cdots \lambda_j^{n_j}\cdots
\lambda_k^{n_k})(i_j-r_j)t_1^{i_1-n_1+r_1}\cdots
t_j^{i_j-n_j+r_j}\cdots t_k^{i_k-n_k+r_k}. $

\end{theorem}

\begin{proof} The proof is completely similar to Theorem
5.3. ~~\end{proof}

 \vspace{4mm}In the following we study the
$3$-Lie algebra $(A, \omega, \delta_{2k})$, where  the derivation
$$\delta_{2k}=t^{2k}\frac{d}{dt}\in Der A, ~k\in
Z.$$

\vspace{2mm}From Lemma 5.2, if $ch F\neq 2$, for
$\delta_{2k}=t^{2k}\frac{d}{dt}\in Der(F[t^{-1}, t])$, $k\in Z$,
then the involution  $\omega:A \rightarrow A$ satisfies ~~
$$\delta_{2k} \omega + \omega \delta_{2k} = 0$$ if and only if $\omega$ is defined as  $
\omega(t^m)=(-1)^m t^m, ~ \forall m\in Z$

Therefore, we have the following result.

\begin{theorem}

If $ch F\neq 2$,
 then  $A$ is a $3$-Lie algebra in the multiplication $[ ,  ,
 ]_{\omega, \delta_{2k}}: $ for arbitrary
 $t^l, t^m, t^n\in A$,

\begin{equation}\hspace{-3.5cm}[t^l, t^m,
t^n]_{\omega, \delta_{2k}}=\begin{vmatrix}
(-1)^lt^{l} &(-1)^mt^{m} &(-1)^nt^{n} \\
\\t^l & t^m & t^n  \\
\\lt^{2k+l-1} & mt^{m+2k-1} & nt^{2k+n-1} \\
\end{vmatrix}
\end{equation}

\vspace{2mm}\hspace{4cm}$=\{(-1)^l(n-m)+(-1)^m(l-n)+(-1)^n(m-l)\}t^{2k+l+m+n-1}.~$

\end{theorem}

\begin{proof} The result follows from Lemma 5.1.
Lemma 5.2 and Theorem 3.3.~~ \end{proof}

Define linear functions $ \alpha,$ $ \beta,$ $ \gamma$ $: A
\rightarrow F$:
$$\alpha(t^m)=(-1)^m,  ~~ \beta(t^m)=1, ~~ \gamma(t^m)=m, ~~\forall t^m\in
A,$$

\vspace{2mm}\noindent that is, for every
$p(t)=\sum\limits_{i=m}^na_it^i\in A$,

$$\alpha(p(t))=\sum\limits_{i=m}^n(-1)^ia_i, ~ ~
\beta(p(t))=\sum\limits_{i=m}^na_i, ~ ~
\gamma(p(t))=\sum\limits_{i=m}^nia_i.$$

Then Eq.(5.7) can be written as:   $\forall $ ~ $t^l, t^m, t^n\in
A,$

$$[t^l, t^m, t^n]_{\omega,\delta_{2k}}=(\alpha\wedge
\beta\wedge\gamma)(t^l, t^m, t^n) t^{l+m+n+2k-1},$$

\vspace{4mm}\noindent where $(\alpha\wedge \beta\wedge\gamma)(t^l,
t^m, t^n)=\begin{vmatrix}
\alpha(t^l)~  & \alpha(t^m) ~ &\alpha(t^n)~  \\
\beta(t^l) ~ & ~ \beta(t^m) & \beta(t^n) ~  \\
\gamma(t^l) ~ & \gamma(t^m) ~ & \gamma(t^n)~  \\
\end{vmatrix}=\begin{vmatrix}
(-1)^l & (-1)^m &(-1)^n \\
1 & 1 & 1  \\
l & m & n \\
\end{vmatrix}.$

\vspace{4mm}\noindent{\bf Remark  }References \cite{ APP, APPS, F2}
studied $n$-Lie algebras $A(G, f, t)$, where $G$ is an additive
Abelian group, $f: G^n\rightarrow F$. By the above discussions, the
$3$-Lie algebra $(A, [ , , ]_{\omega, \delta_{2k}})$ is isomorphic
to the $3$-Lie algebra $A(Z, f, 2k+1)$ in \cite{ APPS}, where $Z$ is
the set of all integers, $f(l, m, n)=\begin{vmatrix}
(-1)^l & (-1)^m &(-1)^n \\
1 & 1 & 1  \\
l & m & n \\
\end{vmatrix}, $ $\forall ~~ l, m, n\in Z.$

\vspace{4mm}\begin{theorem}

Let $A=F[ t^{-1}, t ]$ be the Laurent polynomials over the field of
complex numbers. Then  for any integer $k$, $k\neq 0$, $3$-Lie
algebra $(A, [ , , ]_{\omega,\delta_{2k}})$ with the multiplication
(5.7) is isomorphic to the $3$-Lie algebra $(A, [ , ,
]_{\omega,\delta_0})$, where $\delta_0=\frac{d}{dt}$, and for every
$t^l, t^m, t^n\in A$,
\begin{equation}[t^l, t^m, t^n]_{\omega,
\delta_0}=\{(-1)^l(n-m)+(-1)^m(l-n)+(-1)^n(m-l)\}t^{l+m+n-1}.
\end{equation}

\end{theorem}

\begin{proof} If $k=2s$, define linear mapping $\sigma: (A,
[ , , ]_{\omega, \delta_0}) \rightarrow (A, [ , , ]_{\omega,
\delta_k})$, $\sigma(t^m)=t^{m-k}$, $\sigma(1)=1, $ for every $
t^m\in A.$ Then for every $t^l, t^m, t^n\in A,$
$$\sigma([t^l, t^m, t^n]_{\omega,
\delta_0})=\{(-1)^l(n-m)+(-1)^m(l-n)+(-1)^n(m-l)\}t^{l+m+n-k-1},$$

\vspace{2mm}\noindent$[\sigma(t^l), \sigma(t^m),
\sigma(t^n)]_{\omega, \delta_{2k}}=[t^{l-k}, t^{m-k},
t^{n-k}]_{\omega, \delta_{2k}}$

\vspace{2mm}\hspace{3.3cm}$=\{(-1)^{l-k}(n-m)+(-1)^{m-k}(l-n)+(-1)^{n-k}(m-l)\}t^{l+m+n-k-1}$

\vspace{2mm}\hspace{3.3cm}$=\{(-1)^l(n-m)+(-1)^m(l-n)+(-1)^n(m-l)\}t^{l+m+n-k-1}.$

If $k=2s+1$,  define linear mapping $\sigma: A\rightarrow A$,
$\sigma(t^m)=it^{m-k}$, $\sigma(1)=1 $, where $i^2=-1$. Then for
every $t^l, t^m, t^n\in A,$
$$\sigma([t^l, t^m, t^n]_{\omega,
\delta_0})=i~\{(-1)^l(n-m)+(-1)^m(l-n)+(-1)^n(m-l)\}t^{l+m+n-k-1},$$

\vspace{2mm}\noindent$[ \sigma(t^l),  \sigma(t^m),
\sigma(t^n)]_{\omega, \delta_{2k}}=[ i t^{l-k}, i t^{m-k}, i
t^{n-k}]_{\omega, \delta_{2k}}$

\vspace{2mm}\hspace{3.3cm}$=-i
\{(-1)^{l-k}(n-m)+(-1)^{m-k}(l-n)+(-1)^{n-k}(m-l)\}t^{l+m+n-k-1}$

\vspace{2mm}\hspace{3.3cm}$=i ~
\{(-1)^l(n-m)+(-1)^m(l-n)+(-1)^n(m-l)\}t^{l+m+n-k-1}$

\vspace{2mm}\hspace{3.3cm}$=\sigma([t^l, t^m, t^n]_{\omega,
\delta_0}).$

\vspace{2mm}\noindent The result holds. ~~
\end{proof}

\begin{theorem}

Let $A=F[ t^{-1}, t ]$ over the field $F$ of complex numbers. Then
$3$-Lie algebra $(A, [ , , ]_{\omega, \delta_0})$ is a simple
$3$-Lie algebra.

\end{theorem}

\begin{proof} By Eq.(5.8), for every $t^l, t^m, t^n\in A$,

\begin{equation}[t^l, t^m,
t^{-m+1}]_{\omega,
\delta_0}=\{(-1)^l(-m+1-m)+(-1)^m(l+m-1)+(-1)^{-m+1}(m-l)\}t^l\end{equation}

\hspace{3.9cm}$=
\begin{array}{ll} \left\{\begin{array}{l} 0, ~ if~
l=m ~ or ~ l=-m+1,\\
\{(-1)^{l+1}2m+(-1)^m2l+(-1)^{l}+(-1)^{m+1}\}t^{l} \neq 0, ~ others.
\end{array}\right.
\end{array}$

\vspace{2mm}If $m=0, n=1$, then we have

$[t^l, 1, t^1]_{\omega, \delta_0}=\{(-1)^l+(l-1)+l)\}t^{l}=
\begin{array}{ll} \left\{\begin{array}{l} 0, ~ if ~
l=1, ~ 0,\\
\{2l+(-1)^l-1\}t^l, ~ others.
\end{array}\right.
\end{array}$

\vspace{2mm}Let $I$ be a non-zero ideal of the $3$-Lie algebra $(A,
[ , , ]_{\omega, \delta_0})$.  For every non-zero vector
$p(t)=\sum\limits_{i=r}^sa_it^i\in I$, where $a_s\neq 0, a_r\neq 0,$
and $m$ is a positive integer such that $m
> s$ and $-m+1 < r.$ Thanks to Eq.(5.8) and the Vandermonde
determinant, we have $t^l\in I$ if $a_l\neq 0$ for $r\leq l\leq s$.
We  conclude that there is an integer $l$ such that $t^l\in I.$

Now we prove $I=A$. If $t^m\in I$, then by Eq.(5.9), for every
$l\neq m$ and $l\neq -m+1$,  we have $t^l\in I$. Therefore, we can
choose $j$ satisfying $j\neq \pm m,$  $j\neq -m+1$ (and then
$-m+1\neq -j+1$) such that $t^j\in I$.
 Again by Eq.(5.8), we have $t^m\in I$ and $t^{-m+1}\in I.$

Summarizing above discussions, we get $I=A.$ Therefore, $(A, [ , ,
]_{\omega, \delta_0})$ is a simple $3$-Lie algebra. ~~ \end{proof}

\begin{theorem}

If $ch F=p > 2$, then the $3$-Lie algebra $(A, [ , , ]_{\omega,
\delta_0})$ in Theorem 5.8 is a non-simple $3$-Lie algebra.
\end{theorem}

\begin{proof} Suppose  $I_k=\{ (t^{kp}+t^{-kp})h(t)|~~
\forall h(t)\in A \}$, $J_k=\{ (t^{kp}-t^{-kp})h(t)|~~ \forall
h(t)\in A \}$,  $k\neq 0$. Then $\omega(I_k)\subseteq I_k$, $
\delta_0(I_k)\subseteq I_k$, $\omega(J_k)\subseteq J_k$ and
$\delta_0(J_k)\subseteq J_k$. Thanks to Theorem 3.4, $I_k$ and $J_k$
are non-zero proper ideals of the $3$-Lie algebra $(A, [ , ,
]_{\omega, \delta_0})$. Therefore, the result holds. ~~ \end{proof}

\vspace{2mm}By the above discussions, if  $ch F=p > 2$, then $J_1=\{
(t^{p}-t^{-p})h(t)|~~ \forall h(t)\in A \}$  is an ideal of the
$3$-Lie algebra $(A, [ , , ]_{\omega, \delta_0})$, and satisfies
$\omega(J_1)\subseteq J_1$ and $\delta_0 (J_1)\subseteq J_1$. Then
we get the quotient $3$-Lie algebra of $(A,[ , , ]_{\omega,
\delta_0}) $ relating to the ideal $J_1$, which is denoted by
$(\bar{A},[ , , ]_{\omega, \delta_0})$. The multiplication of
$\bar{A}=A/J_1$ in the basis $\bar{t}^{-p+1}, \cdots \bar{t}^{-1},
\bar{1}, \bar{t}, \cdots, \bar{t}^p$ as follows
\begin{equation}[\bar{t}^{~l}, \bar{t}^{~m}, \bar{t}^{~n}]_{\omega,
\delta_0}=\{(-1)^l(n-m)+(-1)^m(l-n)+(-1)^n(m-l)\}\bar{t}^{~l+m+n-1},
\end{equation}
 where $\bar{t}^{p}=\bar{t}^{-p}.$

\begin{theorem}
The $3$-Lie algebra $(\bar{A}, [ , , ]_{\omega, \delta_0})$ is a
simple $3$-Lie algebra, where  $[ , , ]_{\omega, \delta_0}$ is
defined as Eq. (5.10) and $\dim \bar{A}=2p$.
\end{theorem}

\begin{proof} Let $\bar{I}$ be a nonzero ideal of the
$3$-Lie algebra $(\bar{A}, [ , , ]_{\omega, \delta_0})$. Suppose
$h(\bar{t})=\sum\limits_{i=1-p}^p a_{i}\bar{t}^{i}\in \bar{I}$ and
$h(\bar{t})\neq 0.$

Case I. If $h(\bar{t})=\sum\limits_{i=1-p}^p
a_{i}\bar{t}^{i}=\bar{t}^p$. For every $l$ satisfying $1-p < l < p$,
since $$[\bar{t}^l, \bar{t}^p, \bar{t}^{1-p}]=\{(-1)^l-2 l+1\}
\bar{t}^{l}\in \bar{I}, ~~ (-1)^l-2 l+1\neq 0, ~~ l\neq 1-p,$$

\vspace{2mm}\noindent we get $\bar{t}^l\in \bar{I}$ for $1-p <l\leq
p.$

Thanks to $p\geq 3$, $\bar{t}^2\in \bar{I}$, then $-2
\bar{t}^{p-1}=[\bar{t}^2, \bar{t}^{-1}, \bar{t}^{1-p}]\in \bar{I}$.
It follows $\bar{I}=\bar{A}.$

Case II. If $h(\bar{t})=\sum\limits_{i=1-p}^p a_{i}\bar{t}^{i}\in
\bar{I}$ satisfies $a_p\neq 0$, and there is an integer $k$
satisfying  $1-p\leq k< p$ and $a_k\neq 0$. Without loss of
generality, we suppose  $a_p=1$. By Eq.(5.10)
$$[h(\bar{t}), \bar{t}^{p-1}, \bar{t}^{2-p}]=\sum\limits_{i=1-p}^p
a_{i}(3(-1)^i+2i-1)\bar{t}^{i}\in I.$$

\vspace{2mm}\noindent Since $3(-1)^i+2i-1=0$ if and only if $i=p-1$
or $i=2-p$,  and $3(-1)^i+2i-1=3(-1)^j+2j-1$ if and only if $i=j$
for $i\neq p-1, i\neq 2-p, j\neq p-1, j\neq 2-p$, we obtain
$\bar{t}^p\in I$ (using Vandermonde determinant). Follows the
discussions of the Case I, $ \bar{I}=\bar{A}$.

Case III. If $h(\bar{t})=\sum\limits_{i=1-p}^p
a_{i}\bar{t}^{i}=\sum\limits_{i=1-p}^s a_{i}\bar{t}^{i}$, where  $s
< p$, $a_s=1$ (that is, $a_p=a_{p-1}=\cdots=a_{s+1}=0$) and there is
an integer $k$ satisfying $1-p\leq k < s$ such that $a_k\neq 0$.
Then $s
>1-p$ and
$$[h(\bar{t}), \bar{t}^{p}, \bar{t}^{1-s}]=\sum\limits_{i=1-p}^s
a_{i}\{(-1)^i(1-s)-i-s+1+i(-1)^{s}\}\bar{t}^{i+p-s}=\sum\limits_{j=1-s}^pb_j\bar{t}^j\in
\bar{I},$$

\vspace{2mm}\noindent where
$b_i=a_{i}((-1)^i(1-s)-i-s+1+i(-1)^{s})$. We obtain

\vspace{2mm}\noindent
$b_p=(-1)^s(1-s)-s-s+1+s(-1)^{s}=-2s+(-1)^s+1\neq 0$ since $1-p < s
< p$.

\vspace{2mm}\noindent Follows from Case II, $\bar{I}=\bar{A}.$

Case IV. If $h(\bar{t})=\sum\limits_{i=1-p}^p
a_{i}\bar{t}^{i}=\bar{t}^l$, where $l <  p.$ If $l < p-1$, then

\vspace{2mm}\noindent$[\bar{t}^l, \bar{t}^{p},
\bar{t}^{l+1}]=\{(-1)^l(2l+1)+1\}\bar{t}^p\in \bar{I},$ and
$(-1)^l(2l+1)+1\neq 0,$ we obtain $\bar{t}^p\in \bar{I}.$

\vspace{2mm} If $l=p-1,$ then $[\bar{t}^{p-1}, \bar{t}^{p},
\bar{t}^{-p+2}]=4\bar{t}^p\in \bar{I}.$ Therefore, $\bar{t}^p\in
\bar{I}.$

Summarizing above discussions, $\bar{I}=\bar{A}$. It follows the
result. ~~
\end{proof}

\begin{flushleft}
\section{ conclusions and discussions }
\end{flushleft}
\setcounter{equation}{0}
\renewcommand{\theequation}
{6.\arabic{equation}}

Since the multiple multiplication, constructions of $n$-Lie algebras
is a continuously difficult problem in the structure theory of
$n$-Lie algebras,  for $n\geq 3$.

In \cite{F, D1, APP, APPS}, $n$-Lie algebras are realized by
 associative commutative algebras and its arbitrary $n$ pairwise
commuting derivations, and linear functions.

Papadopoulos in (\cite{G}) constructed $3$-Lie algebras by Dirac
$\gamma$-matrices.  Let $A$ be spanned by the four-dimensional
$\gamma$-matrices $(\gamma^\mu)$ and let
$\gamma^5=\gamma^1\dots\gamma^4$. Then the product
\begin{equation}
[x,y,z]=[[x,y]\gamma^5,z],\;\;\forall a,b,c\in A.
\end{equation}
defines a $3$-Lie algebra which is isomorphic to the unique simple
3-Lie algebra (\cite{F}).

In \cite{HIM}  $3$-Lie algebras are constructed from metric Lie
algebras. Let $(\frak g, B)$ be a metric Lie algebra over a field
$\mathbb F$, that is, $B$ is a nondegenerate symmetric bilinear form
on $\frak g$ satisfying $B([x, y], z)=-B(y, [x, z])$ for every $x,
y, z\in \frak g$.
 Suppose $\{x_1, \cdots, x_m\}$ is a basis of $\frak g$
and $[x_i, x_j]=\sum\limits_{k=1}^m a_{ij}^kx_k, ~1\leq i, j\leq m.$
Set \begin{equation}\frak g_0=\frak g\oplus \mathbb F x^0\oplus
\mathbb F x^{-1}~ \mbox{ (the direct sum of vector
space)}.\end{equation}  Then there is a 3-Lie algebra structure on
$\frak g_0$ given by
\begin{equation}
[x_0, x_i, x_j]=[x_i, x_j], ~ 1\leq i, j\leq m; ~ [x^{-1}, x_i,
x_j]=0, ~ 0\leq i, j\leq m;\end{equation}\begin{equation} [x_i, x_j,
x_k]=\sum\limits_{s=1}^ma_{ij}^sB(x_s, x_k)x^{-1}, ~1\leq i, j,
k\leq m.\end{equation} And $A$ is a metric $3$-Lie algebra in the
multiplication (6.4) and (6.5).

In \cite{BBW}, $3$-Lie algebras are realized by Lie algebras and
linear functions. Let $(L, [ , ])$ be a Lie algebra, $f\in L^*$
satisfying $f([x, y])=0$  for every $x, y\in L$. Then $L$ is a
$3$-Lie algebra in the multiplication
\begin{equation}[x, y, z]_f=f(x)[y, z]+f(y)[z, x]+f(z)[x, y],
\forall x, y, z\in L. \end{equation} And it is proved in \cite{BBW}
that every $m$-dimensional $3$-Lie algebras can be obtained by the
multiplication (6.2) and (6.6) for $m\leq m$.

 Awata, Li and et al in
 \cite{ALMY} constructed a $2$-step solvable $3$-Lie algebra from $(n\times
 n)$-matrices. Let
$\frak g=gl(m,\mathbb F)$ be the general linear Lie algebra. Then
there is a 3-Lie algebra structure on $\frak g$ defined by
\begin{equation}
[A, B, C] = (trA)[B, C] + (trB)[C, A] + (trC)[A, B], \;\;\forall A,
B, C\in \frak g.\end{equation}

In this paper we construct $3$-Lie algebras by associative
commutative algebras and their derivations and involutions. From
Example 4.1 and 4.2 we can obtain $3$-Lie algebras from any
$\alpha\in Hom(G^+, F^+)$ which is not isomorphic to the $3$-Lie
algebra obtain by Eq.(6.7), where $G$ is the set of all $(n\times
n)$-matrices over a field $F$.

So we may provide a problem that is how can we realize $3$-Lie
algebras by Lie algebras, and associative commutative algebras with
multilinear functions and general linear mappings.

\section*{Acknowledgements}
The first author (R. Bai) was supported in part by  the Natural
Science Foundation of Hebei Province (A2010000194).

\bibliography{}

\begin{thebibliography}{999999}



 \bibitem{N} Y. Nambu, Generalized Hamiltonian dynamics, {\it Phys. Rev. D} 7 (1973)
                2405-2412.
\bibitem{T} L. Takhtajan, On foundation of the generalized Nambu mechanics,
{\it Comm. Math. Phys.} 160 (1994) 295-315.

\bibitem{F}  V.T.
Filippov, $n-$Lie algebras,  {\it Sib. Mat.
           Zh.,} 26 (1985) 126-140.



 \bibitem{BL} J. Bagger and N. Lambert, Modeling multiple $M2¡¯$s, Phys. Rev. D 75 (2007) 045020 [hepth/
0611108]; Gauge symmetry and supersymmetry of multiple $M2$-branes,
Phys. Rev. D 77 (2008) 065008 [arXiv:0711.0955]; Comments on
multiple $M2$-branes, JHEP 02 (2008) 105 [arXiv:0712.3738].



\bibitem{HHM} P. Ho, R. Hou and Y. Matsuo, Lie $3$-algebra and
multiple $M_2$-branes, arXiv: 0804.2110.

\bibitem{HCK} P. Ho, M. Chebotar and W. Ke, On skew-symmetric maps on
Lie algebras, {\it Proc. Royal Soc. Edinburgh A} 113 (2003)
1273-1281.

\bibitem{G} A. Gustavsson, Algebraic structures on parallel
M2-branes, arXiv: 0709.1260.

\bibitem{P} G. Papadopoulos, $M2$-branes, $3$-Lie algebras and
Plucker relations, arXiv: 0804.2662.

\bibitem{F2} V.T. Filippov, On $n$-Lie algebra of Jzcpbians, J. Siberian Mathe.,
39(3) (1998): 576-581.

\bibitem{D1} A.S. Dzhumadil¡¯daev,  Identities and derivations for Jacobian
algebras, arXiv: 0202040v1[math. RA]

\bibitem{APP} A. P. Pozhidaev, Monomial $n$-Lie algebras, Algebra
and Logic, 37(5) (1998): 307-322.

\bibitem{APPS} A. P. Pozhidaev, Simple $n$-Lie algebras, Algebra
and Logic, 38(3) (1999): 181-192.






\bibitem{BW2} R. Bai, Y. Wu, J. Li and H. Zhou, Constructing $(n + 1)$-Lie
algebras from n-Lie algebras, J. Phys. A: Math. Theor. 45 (2012)

\bibitem{BBW} R. Bai, C. Bai and J. Wang, Realizations of $3$-Lie
algebras, J. Math. Phys. 51, 063505 (2010). 475206.



\bibitem{KASM} S. Kasymov, On a theory of $n$-Lie algebras, {\it Algebra i Logika} 26 (1987)
277-297.

\bibitem{JTH} J.T. Hartwig, D. Larsson and S. Silvestrov, Deformations of Lie algebras
           using $\sigma$-deivations, arXiv: 5036 2003

   \bibitem{HIM}   P. Ho, Y. Imamura, Y. Matsuo, $M2$ to $D2$
   revisited, {\it J. High Energy Phys.} 0807,003 2008; e-print arXiv:0805.1202.

 \bibitem{ALMY} H. Awata, M. Li, D. Minic and T. Yoneya, On the quantization of Nambu brackets, {\it JHEP}  2 (2001) 013 (17pp).




\end{thebibliography}

\end{document}